\documentclass[letterpaper]{IEEEtran}

\ifCLASSINFOpdf
\else
   \usepackage[dvips]{graphicx}
\fi
\usepackage{url}

\usepackage{graphicx}

\usepackage{macros}
\usepackage{cleveref}

\begin{document}

\title{On the Real
Part of Butterfly Factorizations,
\\with Application to Fast Transforms}

\author{\IEEEauthorblockN{Maël~Chaumette, Rémi~Gribonval and Elisa~Riccietti}%
\thanks{Manuscript received ...}%
\thanks{This project was supported by the SHARP project of the PEPR-IA (ANR-23-PEIA-0008, granted by France 2030), and the project ANR MEPHISTO ANR-24-CE23-7039.}%
\thanks{Maël Chaumette and Rémi Gribonval are with Inria, CNRS, ENS de Lyon, Université Claude Bernard Lyon 1, LIP, UMR 5668, 69342, France. (e-mail: mael.chaumette@inria.fr; remi.gribonval@inria.fr)}%
\thanks{Elisa Riccietti is with ENS de Lyon, CNRS, Inria, Université Claude Bernard Lyon 1, LIP, UMR 5668, 69342, France. (e-mail: elisa.riccietti@ens-lyon.fr)}}

\markboth{IEEE Signal Processing Letters,~Vol.~XX, No.~XX, 2026}%
{Shell \MakeLowercase{\textit{et al.}}: Bare Demo of IEEEtran.cls for IEEE Journals}
\maketitle

\begin{abstract}
    Fast transforms such as the Discrete Cosine Transform can be expressed as the real part of a product of complex-valued structured sparse matrices known as Kronecker-sparse factors.
    More generally, products of such structured sparse matrices appear in square dyadic butterfly and monarch matrices, which are known to represent many fast transforms efficiently.
    Yet, it has remained unclear whether taking the real or the imaginary part of such products preserves this structure.
    In this letter, we give a condition on the factors to ensure that the structure is preserved.
    This condition is satisfied in the case of the Discrete Cosine Transform, the Discrete Sine Transform and the Discrete Hartley Transform, allowing us to provide a new representation of these transforms as products of real-valued Kronecker-sparse factors, which opens new avenues for efficient GPU implementations using low-precision numerical formats.
\end{abstract}



\section{Introduction}

\IEEEPARstart{F}{ast} transforms are a fundamental tool in signal processing and machine learning, enabling efficient computation of linear operations.
A given linear operator is usually made fast by expressing it as a product of sparse factors, which can be efficiently stored and implemented.
The use of such sparse factorizations dates back to the Fast Fourier Transform (FFT) algorithm \cite{cooley1965algorithm}, which can be interpreted as a factorization into factors with a so-called {\em butterfly} structure, reducing costs from quadratic to near-linear in the matrix dimension.
Butterfly structures also appear in fast transforms such as the Discrete Cosine/Sine/Hartley Transform (DCT/DST/DHT) \cite{ahmed1974discrete,makhoul1980Fast,bracewell1984fast,sorensen1985computing}, or the Fast Noiselet transform \cite{coifman2001noiselets}, as well as in hierarchical methods \cite{michielssen1996multilevel}, in the fast evaluation of Fourier integral operators \cite{candes2009fast}, in learning or approximating linear operators \cite{le2016flexible} and in deep learning \cite{dao2019learning}, \cite{dao2020kaleidoscope}, \cite{chen2021pixelated}.
Recently, a general class of structured sparse matrices, called Kronecker-sparse factors \cite{le2025butterfly}, introduced to generalize the butterfly factorization \cite{li2015butterfly}, has been shown to express many families of efficient matrices (see Example \ref{ex:ksf}) while being endowed with provably good factorization algorithms.

The benefits of sparse factorizations in terms of computation time are well-known on CPUs \cite{puschel2005spiral}, \cite{le2016flexible}.
On GPUs, the situation is more subtle but rapidly evolving: while irregular sparsity patterns are not well-suited and make the theoretical speed-up difficult to achieve in practice \cite{chen2021pixelated}, \cite{hoefler2021sparsity}, recent works with block structures or dedicated kernels start to close this gap \cite{dao2022monarch}, \cite{gonon2025fast}.
Furthermore, the availability of low-precision \texttt{float} formats on GPUs offers additional avenues for reducing memory footprint and improving computational efficiency.

Despite becoming standard, low-precision formats on GPUs are however currently not available for complex numbers.
As the implementation of many (real-valued) fast transforms such as the DCT, DST or DHT relies on their expression using the real (or the imaginary) part of the FFT, this involves complex-valued Kronecker-sparse factors which cannot benefit from standard low-precision format.
This raises the question of whether the real (or imaginary) part of Kronecker-sparse factorizations is still a Kronecker-sparse factorization {\em with real-valued entries}.
To the best of our knowledge, this has not been addressed in the literature.
The contribution of this letter is, under a choice of permutation matrices, to give a necessary and sufficient condition on the factors to ensure that the real or the imaginary part of a product of Kronecker-sparse factors is still a product of Kronecker-sparse factors.
This allows us to provide a new representation of the DCT, DST and DHT as a factorization with real-valued sparse matrices, unlocking the use of low-precision format.

\section{Notations and background}

We denote by $\Re(\cdot)$ and $\Im(\cdot)$ the real and imaginary parts of a complex number, respectively.
Vectors (resp. matrices) are denoted by lowercase (resp. uppercase) bold letters.
$\mOne_{m \times n}$ denotes the matrix of size $m \times n$ with all entries equal to $1$.
$\mI_n$ denotes the identity matrix of size $n \times n$.
Finally, for $a \leq b \in \bbN$, we denote by $\llbracket a, b \rrbracket$ the set of integers $\{ a, a+1, \ldots, b \}$.

Kronecker-sparse factors are defined by their support, i.e., the set of indices where the matrix has non-zero entries.

\begin{definition}[Kronecker-sparse supports and their (sparsity) patterns {\cite[Definition~4.1]{le2025butterfly}}]\label{def:ksf}
    For $a, b, c, d \in \bbN$, a {\em Kronecker-sparse factor} of (sparsity) {\em pattern} $\vpi := (a, b, c, d)$ is a matrix in $\bbR^{a b d \times a c d}$ or $\bbC^{a b d \times a c d}$ such that its support is included in the support of $\mS_{\vpi} := \mI_a \otimes \mOne_{b \times c} \otimes \mI_d \in \{ 0, 1 \}^{a b d \times a c d}$.
    The set of Kronecker-sparse factors with pattern $\vpi$ is denoted by $\cK_{\vpi}$.
\end{definition}

\figurename~\ref{fig:action-perm} shows two examples of Kronecker-sparse factors.

\begin{example}\label{ex:ksf}
    Such factors are the basic blocks to build many matrices with efficient implementations:
    \begin{itemize}
        \item {\em Square dyadic butterflies} \cite{dao2019learning,zheng2023efficient} are related to many fast transforms, such as the FFT \cite{cooley1965algorithm} and the Fast Walsh-Hadamard Transform \cite{fino1976unified}, which are $n\times n$ matrices with $n=2^L$ that can be written in terms of a product $\mK_{1}\ldots \mK_{L}$ where $\mK_{\ell} \in \cK_{\vpi_{\ell}}$, $\ell \in \llbracket 1, L \rrbracket$ and $\vpi_\ell =(2^{\ell - 1}, 2, 2, 2^{L - \ell})$.
        \item {\em Monarch factorizations}  \cite{dao2022monarch} for weight matrices in transformers (for vision or language tasks) have led to speed-ups of training and inference time. They express an $m \times n$ matrix
        as a product $\mK_{1} \mK_{2}$ with $\mK_{\ell} \in \cK_{\vpi_{\ell}}$, $\ell =1,2$, where
        $\vpi_1 = (1, p, q, m / p)$ and $\vpi_2 = (q, m / p, n / q, 1)$ assuming that $p, q$ divide $m, n$ respectively.
    \end{itemize}
    General matrix products $\mK_{1}\cdots \mK_{L}$ with $\mK_{\ell} \in\cK_{\vpi_{\ell}}$, $\ell \in \llbracket 1, L \rrbracket$ for some patterns $\vpi_{\ell}$ are called {\em butterfly factorizations} \cite{zheng2023efficient} and $(\vpi_1, \ldots, \vpi_L)$ is called the {\em architecture} of the factorization.
\end{example}

\section{Main result}
\label{sec:main}

In this section, we will prove\footnote{All definitions and proofs of this letter were formalized in the Lean proof assistant. The code is available in the supplementary material.} that, under some condition on the architecture, if a matrix admits a butterfly factorization $\mA=\mK_{1}\ldots \mK_{L}$, then the matrix $\mA' := \Re(\mA)$ also admits a butterfly factorization with a different architecture.
We first show that $\Re(\mA)$ is expressed as a product of $L$ matrices {\em without} the Kronecker-sparse structure, then we apply well-chosen permutations to recover a Kronecker-sparse factorization.

\subsection{Pre-factorization of the real part}

When $L=1$ (i.e., $\mA := \mK_1$), as being a Kronecker-sparse factor only depends on the support of $\mA$, $\Re(\mA)$ and $\Im(\mA)$ are also Kronecker-sparse factors with the same pattern.
For $L \geq 2$, simple linear algebra yields $\Re(\mA) = \overline{\mK}_1 \cdots \overline{\mK}_L$ where
\begin{equation}\label{eq:concat_ksf}
    \begin{split}
        &\overline{\mK}_1 := \begin{pmatrix}
            \Re(\mK_1) & \Im(\mK_1)
        \end{pmatrix}\\
        &\overline{\mK}_\ell := \begin{pmatrix}
            \Re(\mK_\ell) & \Im(\mK_\ell)\\
            -\Im(\mK_\ell) & \Re(\mK_\ell)
        \end{pmatrix}, \quad 2 \leq \ell \leq L - 1\\
        &\overline{\mK}_L := \begin{pmatrix}
            \Re(\mK_L)\\
            -\Im(\mK_L)
        \end{pmatrix}\,.
    \end{split}
\end{equation}
However, these factors $\overline{\mK}_\ell$, $\ell \in \llbracket 1, L \rrbracket$, are not Kronecker-sparse as illustrated in the middle of \figurename~\ref{fig:action-perm} with the support of $\overline{\mK}_{L}$ when $\vpi_{L}=(2,2,2,3)$.

\subsection{Permutations to recover a Kronecker-sparse factorization}
Considering any permutation matrices $\mP_2, \ldots, \mP_L$, we have
\begin{equation*}
    \overline{\mK}_1 \cdots \overline{\mK}_L = \widetilde{\mK}_1 \cdots \widetilde{\mK}_L
\end{equation*}
with the permuted factors
\begin{equation}\label{eq:correct_ksf}
    \begin{split}
        &\widetilde{\mK}_1 := \overline{\mK}_1 \mP_2^\top\\
        &\widetilde{\mK}_\ell := \mP_\ell \overline{\mK}_\ell \mP_{\ell + 1}^\top, \quad 2 \leq \ell \leq L - 1\\
        &\widetilde{\mK}_L := \mP_L \overline{\mK}_L\,.
    \end{split}
\end{equation}
Can we find permutations and ``maximally sparse'' patterns $\vpi'_{\ell}, 1 \leq \ell \leq L$ such that $\widetilde{\mK}_{\ell} \in \cK_{\vpi'_{\ell}}$, $1 \leq \ell \leq L$?
This is what we investigate next. Quite naturally, based on a count of the number of nonzero entries per row/column of each factor (for example, $\overline{\mK}_1$ has twice the number of columns of $\mK_{1}$, and the sparsity of each column is unchanged), we end up targeting the following sparsity patterns:
\begin{equation}\label{eq:new_chain}
    \vpi'_\ell := \begin{cases}
        (a_1, b_1, 2 c_1, d_1) & \text{if } \ell = 1 \\
        (a_L, 2 b_L, c_L, d_L) & \text{if } \ell = L \\
        (a_\ell, 2 b_\ell, 2 c_\ell, d_\ell) & \text{otherwise.}
    \end{cases}
\end{equation}

\subsection{Finding the right permutation matrices}
To satisfy $\widetilde{\mK}_{\ell} \in \cK_{\vpi'_{\ell}}$ {\em irrespective of the values of the entries} of $\mK_1, \ldots, \mK_L$, the permutation matrices $\mP_\ell$ have to satisfy:
\begin{align}
    &\mP_2 \cdot \begin{pmatrix}
        \mS_{\vpi_1}^\top\\
        \mS_{\vpi_1}^\top
    \end{pmatrix} = \mS_{(a_1, b_1, 2 c_1, d_1)}^\top = \mS_{(a_1, 2 c_1, b_1, d_1)} \label{eq:cond1}\\
    &\mP_\ell \cdot \begin{pmatrix}
        \mS_{\vpi_\ell} & \mS_{\vpi_\ell}\\
        \mS_{\vpi_\ell} & \mS_{\vpi_\ell}
    \end{pmatrix} \cdot \mP_{\ell + 1}^\top = \mS_{(a_\ell, 2 b_\ell, 2 c_\ell, d_\ell)}, \quad 2 \leq \ell \leq L - 1 \label{eq:cond2}\\
    &\mP_L \cdot \begin{pmatrix}
        \mS_{\vpi_L}\\
        \mS_{\vpi_L}
    \end{pmatrix} = \mS_{(a_L, 2 b_L, c_L, d_L)}\,. \label{eq:cond3}
\end{align}
We will use a particular choice of permutation matrices to satisfy \eqref{eq:cond1}, \eqref{eq:cond2} and \eqref{eq:cond3}.
Potentially, other choices of permutation matrices could also satisfy these conditions, but we do not explore this question in this letter.
We propose to use $\mP_{\ell} := \mP_{\vpi_{\ell}}$ where, given any pattern $\vpi = (a, b, c, d)$, we define
\begin{equation}\label{eq:permutation}
    \mP_{\vpi} := \mC_{a, 2} \otimes \mI_{b d}
\end{equation}
with $\mC_{m, n} \in \bbR^{m n \times m n}$ the so-called  {\em commutation matrix} \cite{magnus1979commutation}, also known as {\em $(m,n)$-perfect shuffle} \cite{van2000ubiquitous}.

\begin{definition}[Commutation matrix {\cite{magnus1979commutation}}]
    The commutation matrix $\mC_{m,n}$ is the permutation matrix such that, for any matrix $\mM \in \bbR^{m \times n}$, $\mC_{m, n} \vect(\mM) = \vect(\mM^\top)$, where $\vect(\cdot)$ denotes the vectorization operator stacking the columns of $\mM$. It corresponds to the following permutation $\sigma$ of the set $\llbracket 1, m n \rrbracket$
    \begin{equation*}
        \sigma(i + m(j - 1)) = j + n (i - 1), \quad i \in \llbracket 1, m \rrbracket, j \in \llbracket 1, n \rrbracket\,.
    \end{equation*}
\end{definition}

\begin{figure}[!t]
    \centering
    \includegraphics[width=0.9\columnwidth]{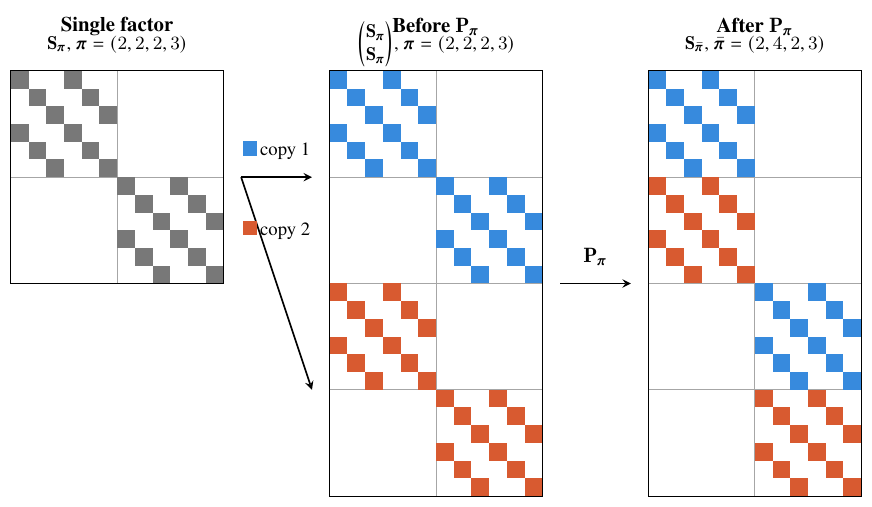}
    \caption{({\em Left}) Example of the support of a Kronecker-sparse factor with pattern $\vpi = (2, 2, 2, 3)$.
    Here, the non-zero entries are represented by the colored squares.
    ({\em Middle}) The support of the concatenated Kronecker-sparse factors $\begin{pmatrix} \mS_{\vpi}\\ \mS_{\vpi} \end{pmatrix}$, which is not a Kronecker-sparse factor.
    ({\em Right}) The support of the Kronecker-sparse factor with pattern $\overline{\vpi} = (2, 4, 2, 3)$ obtained by applying the permutation matrix $\mP_{\vpi}$ on the left to the previous matrix.}
    \label{fig:action-perm}
\end{figure}

In particular $\mC_{a, 2}$ applied to a matrix permutes its rows by grouping on top the rows of indices $i = 0$ and $i = a$, then the rows of indices $i = 1$ and $i = a + 1$, and so on.
By construction, the matrix $\mP_{\vpi}$ acts similarly but on groups of $b d$ consecutive rows instead of single rows.
Hence, applying $\mP_{\vpi}$ on the left to the concatenated Kronecker-sparse factors appearing in \eqref{eq:cond1}, \eqref{eq:cond2} and \eqref{eq:cond3} with pattern $\vpi$ will group the rows to obtain a Kronecker-sparse factor with pattern $\overline{\vpi} := (a, 2 b, c, d)$.
This is illustrated in \figurename~\ref{fig:action-perm} for $\vpi = (2, 2, 2, 3)$ and formalized in the following lemma.
\begin{lemma}\label{lem:effect_perm}
    Given $\vpi = (a, b, c, d)$, the matrix $\mP_{\vpi}$ satisfies
    \begin{equation*}
        \mP_{\vpi} \cdot \begin{pmatrix}
            \mS_{\vpi}\\
            \mS_{\vpi}
        \end{pmatrix} = \mS_{\overline{\vpi}} \in \bbR^{2 a b d \times a c d}\,,
    \end{equation*}
    where $\overline{\vpi} := (a, 2 b, c, d)$.
\end{lemma}

\begin{IEEEproof}
    A key property \cite[Theorem~3.1]{magnus1979commutation} of $\mC_{m,n}$ is the identity $\mA \otimes \mB = \mC_{m,n}(\mB \otimes \mA)\mC_{p,q}^{\top}$ for every $\mA \in \bbC^{m \times p}$ and $\mB \in \bbC^{n \times q}$, with $\otimes$ the Kronecker product.
    Since $\mC_{p,1}=\mI_{p}$ for every $p$ we deduce when specializing to $m=a$, $n=2$ and $q=1$ that
    \begin{equation}\label{eq:CommKron}
        \mA \otimes \mB = \mC_{a,2}(\mB \otimes \mA), \quad
        \forall \mA \in \bbC^{a \times p}, \mB \in \bbC^{2 \times 1}\,.
    \end{equation}
    Combined with~\Cref{def:ksf}, \eqref{eq:permutation} and the mixed-product property ($(\mA \otimes \mB) (\mC \otimes \mD) = (\mA \mC) \otimes (\mB\mD)$ for appropriate dimensions  \cite{van2000ubiquitous}) of the Kronecker product, we obtain
    \begin{IEEEeqnarray*}{rCl+x*}
        \mP_{\vpi} \cdot \begin{pmatrix}
            \mS_{\vpi}\\
            \mS_{\vpi}
        \end{pmatrix}
            &=& (\mC_{a, 2} \otimes \mI_{b d}) \cdot (\mOne_{2 \times 1} \otimes \mI_a \otimes \mOne_{b \times c} \otimes \mI_d)\\
            &=& (\mC_{a, 2} \cdot (\mOne_{2 \times 1} \otimes \mI_a)) \otimes (\mI_{b d} (\mOne_{b \times c} \otimes \mI_d))\\
            &=& (\mI_a \otimes \mOne_{2 \times 1}) \otimes (\mOne_{b \times c} \otimes \mI_d)\\
            &=& \mI_a \otimes \underbrace{(\mOne_{2 \times 1} \otimes \mOne_{b \times c})}_{= \mOne_{2 b \times c}} \otimes \mI_d = \mS_{\overline{\vpi}}\,. \hfill \IEEEQEDhere
    \end{IEEEeqnarray*}
\end{IEEEproof}

We now characterize when $\mP_{\ell}:=\mP_{\vpi_{\ell}}$ satisfy \eqref{eq:cond1}-\eqref{eq:cond2}-\eqref{eq:cond3}.

\begin{lemma}
    For any architecture $(\vpi_1, \ldots, \vpi_L)$ with $\vpi_\ell = (a_\ell, b_\ell, c_\ell, d_\ell)$, $\ell \in \llbracket 1, L \rrbracket$, \eqref{eq:cond3} is satisfied with $\mP_L := \mP_{\vpi_L}$.
    Besides, \eqref{eq:cond1} and \eqref{eq:cond2} hold with $\mP_\ell := \mP_{\vpi_\ell}$ if, and only if, for every $\ell \in \llbracket 1, L - 1 \rrbracket$, we have
    \begin{equation}\label{eq:cond_final}
        \begin{split}
            &(\mC_{a_{\ell + 1}, 2} \otimes \mI_{b_{\ell + 1} d_{\ell + 1}}) \cdot\\
            &\quad\left( \mOne_{2 \times 1} \otimes \mI_{a_\ell} \otimes \mOne_{c_\ell \times b_\ell'} \otimes \mI_{d_\ell} \right) = \mI_{a_\ell} \otimes \mOne_{2 c_\ell \times b_\ell'} \otimes \mI_{d_\ell}\,,
        \end{split}
    \end{equation}
    where $b_\ell' = b_\ell$ if $\ell = 1$ and $b_\ell' = 2 b_\ell$ otherwise.
\end{lemma}

\begin{IEEEproof}
    First, as a consequence of Lemma~\ref{lem:effect_perm}, we have
    \begin{equation*}
        \mP_{\vpi_\ell} \cdot \begin{pmatrix}
            \mS_{\vpi_\ell}\\
            \mS_{\vpi_\ell}
        \end{pmatrix} = \mS_{\overline{\vpi}_\ell}, \quad 2 \leq \ell \leq L\,,
    \end{equation*}
    so \eqref{eq:cond3} holds with $\mP_L := \mP_{\vpi_L}$.
    Next, \eqref{eq:cond1} rewrites (for $\ell = 1$) as
    \begin{equation*}
        \mP_{\vpi_{\ell + 1}} \cdot \begin{pmatrix}
            \mS_{(a_\ell, c_\ell, b_\ell, d_\ell)}\\
            \mS_{(a_\ell, c_\ell, b_\ell, d_\ell)}
        \end{pmatrix} = \mS_{(a_\ell, 2 c_\ell, b_\ell, d_\ell)}\,,
    \end{equation*}
    which is equivalent to \eqref{eq:cond_final} for $\ell=1$ by~\Cref{def:ksf} and \eqref{eq:permutation}.
    Now, denoting $\bar{\overline{\vpi}} := (a, 2 b, 2 c, d)$ for $\vpi := (a, b, c, d)$, since
    \begin{equation*}
        \mP_{\vpi_\ell} \cdot \begin{pmatrix}
            \mS_{\vpi_\ell} & \mS_{\vpi_\ell}\\
            \mS_{\vpi_\ell} & \mS_{\vpi_\ell}
        \end{pmatrix} = \begin{pmatrix}
            \mS_{\overline{\vpi}_\ell} & \mS_{\overline{\vpi}_\ell}
        \end{pmatrix}, \quad 2 \leq \ell \leq L\,,
    \end{equation*}
    we get that \eqref{eq:cond2} is equivalent to $\begin{pmatrix}
            \mS_{\overline{\vpi}_\ell} & \mS_{\overline{\vpi}_\ell}
        \end{pmatrix} \cdot \mP_{\vpi_{\ell + 1}}^\top = \mS_{\bar{\overline{\vpi}}_\ell}, \quad 2 \leq \ell \leq L - 1\,,$ which can be rewritten as
    \begin{equation*}
        \mP_{\vpi_{\ell + 1}} \cdot \begin{pmatrix}
            \mS_{(a_\ell, c_\ell, 2 b_\ell, d_\ell)}\\
            \mS_{(a_\ell, c_\ell, 2 b_\ell, d_\ell)}
        \end{pmatrix} = \mS_{(a_\ell, 2 c_\ell, 2 b_\ell, d_\ell)}, \quad 2 \leq \ell \leq L - 1\,,
    \end{equation*}
    or equivalently as \eqref{eq:cond_final} for $2 \leq \ell \leq L-1$.
   \end{IEEEproof}

This lemma turns conditions \eqref{eq:cond1}, \eqref{eq:cond2}, and \eqref{eq:cond3} on the permutation matrices into \eqref{eq:cond_final}, a condition {\em on the architecture} of the Kronecker-sparse factors.
To conclude, we turn this condition into three explicit conditions on the Kronecker-sparse patterns.

\begin{proposition}\label{prop:existFacto}
    Consider an architecture $(\vpi_1, \ldots, \vpi_L)$. The following are equivalent:
    \begin{itemize}
        \item for every $\ell \in \llbracket 1, L - 1 \rrbracket$ we have
            \begin{itemize}[\IEEEsetlabelwidth{C1.}]
                \item[C1.] $a_\ell c_\ell d_\ell = a_{\ell + 1} b_{\ell + 1} d_{\ell + 1}$\label{item:cond1};
                \item[C2.] $d_{\ell} | b_{\ell + 1} d_{\ell + 1}$\label{item:cond2};
                \item[C3.] $a_\ell | a_{\ell + 1}$\label{item:cond3}.
            \end{itemize}
        \item defining $\widetilde{\mK}_{\ell}$ as in~\eqref{eq:correct_ksf} with $\mP_\ell := \mP_{\vpi_\ell}$ as in \eqref{eq:permutation} and $\mK_{\ell} \in \cK_{\vpi_{\ell}}$, we have $\widetilde{\mK}_{\ell} \in \cK_{\vpi'_{\ell}}$, $1 \leq \ell \leq L$  with $\vpi_\ell'$ as in \eqref{eq:new_chain}.
    \end{itemize}
\end{proposition}

\begin{IEEEproof}
    C1 is necessary and sufficient to ensure compatible matrix dimensions for matrix multiplication.\\
    {\em Sufficiency of C1-C2-C3.}
    We denote $r := \frac{b_{\ell + 1} d_{\ell + 1}}{d_\ell}$, which by C2 is an integer and by C1 is equal to $\frac{a_\ell c_\ell}{a_{\ell + 1}}$.
    Moreover, as $a_\ell | a_{\ell + 1}$ by C3, there exists an integer $p$ such that $p := \frac{a_{\ell + 1}}{a_\ell}$, so $c_\ell = p r$.
    Hence, we can decompose the matrix $\mOne_{c_\ell \times b_\ell'} = \mOne_{p \times b_\ell'} \otimes \mOne_{r \times 1}$ then use the mixed-product property of the Kronecker product, therefore the left-hand side of \eqref{eq:cond_final} can be rewritten as
    \begin{align*}
        &(\mC_{a_{\ell + 1}, 2} \otimes \mI_{d_{\ell + 1} b_{\ell + 1}}) \cdot
        \left( \mOne_{2 \times 1} \otimes \mI_{a_\ell} \otimes \mOne_{c_\ell \times b_\ell'} \otimes \mI_{d_\ell} \right)\\
            &= (\mC_{a_{\ell + 1}, 2} \otimes \mI_{d_{\ell + 1} b_{\ell + 1}}) \cdot \left( \mOne_{2 \times 1} \otimes \mI_{a_\ell} \otimes \mOne_{p \times b_\ell'} \otimes \mOne_{r \times 1} \otimes \mI_{d_\ell} \right)\\
            &= \mC_{a_{\ell + 1}, 2} \cdot \left( \mOne_{2 \times 1} \otimes \mI_{a_\ell} \otimes \mOne_{p \times b_\ell'} \right) \otimes \left( \mI_{d_{\ell + 1} b_{\ell + 1}} \cdot (\mOne_{r \times 1} \otimes \mI_{d_\ell}) \right)\\
            &= \mI_{a_\ell} \otimes \mOne_{p \times b_\ell'} \otimes \mOne_{2 \times 1} \otimes \mOne_{r \times 1} \otimes \mI_{d_\ell}\\
            &= \mI_{a_\ell} \otimes \mOne_{2 p \times b_\ell'} \otimes \mOne_{r \times 1} \otimes \mI_{d_\ell}
            = \mI_{a_\ell} \otimes \mOne_{2 c_\ell \times b_\ell'} \otimes \mI_{d_\ell}\,.
    \end{align*}

    {\em Necessity of C2\&C3.}
    With $\mX_\ell := \mOne_{2 \times 1} \otimes \mI_{a_\ell} \otimes \mOne_{c_\ell \times b_\ell'} \otimes \mI_{d_\ell}$, $\mY_\ell := \mI_{a_\ell} \otimes \mOne_{2 c_\ell \times b_\ell'} \otimes \mI_{d_\ell}$, $\ell \in \llbracket 1, L - 1 \rrbracket$, \eqref{eq:cond_final} reads $\mP_{\vpi_{\ell + 1}} \cdot \mX_\ell = \mY_\ell$.
    We continue by indexing the rows of $\mX_\ell$ by $(\varepsilon, i, k_1, k_2) \in \{0, 1 \} \times \llbracket 0, a_\ell - 1 \rrbracket \times \llbracket 0, c_\ell - 1 \rrbracket \times \llbracket 0, d_\ell - 1 \rrbracket$ via $\mathtt{idx}(\varepsilon, i, k_1, k_2) := \varepsilon (a_\ell c_\ell d_\ell) + i (c_\ell d_\ell) + k_1 d_\ell + k_2$, and the columns of $\mX_\ell$ (equivalently of $\mY_\ell$ because they live in the same column space of size $a_\ell b_\ell' d_\ell$) by an index defined using $(i', j, k_2') \in \llbracket 0, a_\ell - 1 \rrbracket \times \llbracket 0, b_\ell' - 1 \rrbracket \times \llbracket 0, d_\ell - 1 \rrbracket$ similarly.
    Since $\mX_\ell = \mOne_{2 \times 1} \otimes \mI_{a_\ell} \otimes \mOne_{c_\ell \times b_\ell'} \otimes \mI_{d_\ell}$
    \begin{equation*}
        \mX_\ell[(\varepsilon, i, k_1, k_2), (i', j, k_2')] = \delta_{i, i'} \delta_{k_2, k_2'}\,.
    \end{equation*}
    We now index the rows of $\mY_\ell$ by $(i, k_1', k_2) \in \llbracket 0, a_\ell - 1 \rrbracket \times \llbracket 0, 2 c_\ell - 1 \rrbracket \times \llbracket 0, d_\ell - 1 \rrbracket$ via $\mathtt{idx}'(i, k_1', k_2) := i (2 c_\ell d_\ell) + k_1' d_\ell + k_2$.
    Since $\mY_\ell = \mI_{a_\ell} \otimes \mOne_{2 c_\ell \times b_\ell'} \otimes \mI_{d_\ell}$, we have
    \begin{equation*}
        \mY_\ell[(i, k_1', k_2), (i', j, k_2')] = \delta_{i, i'} \delta_{k_2, k_2'}\,.
    \end{equation*}
    In particular, to each fixed $(i, k_2)$ we can associate a {\em distinct} row vector $\v{v}_{i,k_{2}}$ such that all rows of $\mX_\ell$ with second/last coordinates $(i, k_2)$ (there are $2 c_\ell$ of them indexed by $(\varepsilon, k_1)$) and all rows of $\mY_\ell$ with first/last coordinates $(i, k_2)$ (there are $2 c_\ell$ of them indexed by $k_1'$) are equal to $\v{v}_{i,k_{2}}$.
    Hence, denoting by $\rho_{\vpi}$ the row-permutation associated to $\mP_{\vpi}$, the equality $ \mP_{\vpi_{\ell + 1}} \mX_\ell = \mY_\ell$ implies (and is in fact equivalent to)
    \begin{equation}\label{eq:equivalence}
        \begin{split}
            &\rho_{\vpi_{\ell + 1}}\left( \{ \mathtt{idx}(\varepsilon, i, k_1, k_2) \}_{\varepsilon \in \{0, 1\}, k_1 \in \llbracket 0, c_\ell - 1 \rrbracket} \right)\\
            &= \{ \mathtt{idx}'(i, k_1', k_2) \}_{k_1' \in \llbracket 0, 2 c_\ell - 1 \rrbracket},\\
            &\forall (i, k_2) \in \llbracket 0, a_\ell - 1 \rrbracket \times \llbracket 0, d_\ell - 1 \rrbracket\,.
        \end{split}
    \end{equation}

    To prove C2, we use \eqref{eq:equivalence} with $i=k_2 = 0$.
    As the integers in $\{ \mathtt{idx}'(0, k_1', 0) \}_{k_1' \in \llbracket 0, 2 c_\ell - 1 \rrbracket}$ are multiples of $d_\ell$, $d_{\ell}$ divides $\rho_{\vpi_{\ell + 1}}\left( \mathtt{idx}(\varepsilon, 0, k_1, 0) \right)$. For $\varepsilon =1, k_1=0$, since  $\mathtt{idx}(1, 0, 0, 0) = a_\ell c_\ell d_\ell = a_{\ell + 1} b_{\ell + 1} d_{\ell + 1}$ (by C1), this means that $d_{\ell}$ divides $\rho_{\vpi_{\ell + 1}}\left( a_{\ell + 1} b_{\ell + 1} d_{\ell + 1}\right)$.
    By definition of $\mP_{\vpi_{\ell+1}}$ we have $\rho_{\vpi_{\ell + 1}}(a_{\ell + 1} b_{\ell + 1} d_{\ell + 1}) = b_{\ell + 1} d_{\ell + 1}$, this establishes C2.

    We use \eqref{eq:equivalence} with $i=1,\varepsilon=k_{1}=k_2 = 0$.
    As $\mathtt{idx}(0, 1, 0, 0) = c_\ell d_\ell$, we obtain as above that $\rho_{\vpi_{\ell + 1}}(c_{\ell}d_{\ell}) = (2c_{\ell}+k'_{1}) d_{\ell}$ for some $k_1' \in \llbracket 0, 2 c_\ell - 1 \rrbracket$.
    By definition of $\mP_{\vpi_{\ell+1}}$, $\rho_{\vpi_{\ell + 1}}(c_{\ell}d_{\ell})$ depends on the  Euclidean division of $c_{\ell}d_{\ell}$ by $b_{\ell+1}d_{\ell+1}$.
    By C2, $r:= b_{\ell+1}d_{\ell+1}/d_{\ell}$ is an integer, and using the Euclidean division $c_{\ell} = pr+q$, $p,q \in \bbN$, $0 \leq q<r$, with C1, we get $c_\ell d_\ell = p b_{\ell + 1} d_{\ell + 1} + q d_\ell$ and $q d_\ell < r d_\ell = b_{\ell + 1} d_{\ell + 1}$.
    It follows that $\rho_{\vpi_{\ell + 1}}(c_\ell d_\ell) = (2 p r + q) d_\ell$, hence $
    (2c_{\ell}+k'_{1})d_{\ell}=(2pr+q)d_{\ell}$, {\em i.e.}, $2c_{\ell}+k'_{1}=2pr+q=2(c_{\ell}-q)+q=2c_{\ell}-q$. Thus, $k'_{1}=q=0$, hence $r | c_{\ell}$, and since by C1 $r = a_{\ell}c_{\ell}/a_{\ell+1}$, we conclude that $a_{\ell} | a_{\ell+1}$.
    This establishes C3.\hfill \IEEEQEDhere
\end{IEEEproof}

\subsection{Consequences}
A first direct consequence of this proposition is the following corollary showing that every chainable architecture\footnote{This condition is sufficient to ensure that the product of two Kronecker-sparse factors is also a Kronecker-sparse factor. This condition corresponds to $\text{C1} \land \text{C2} \land \text{C3} \land (d_{\ell + 1} | d_\ell,\, \forall \ell)$.} satisfies the three conditions of the proposition.

\begin{corollary}
    If a matrix $\mA$ admits a Kronecker-sparse factorization with a chainable architecture \cite[Definition~4.12]{le2025butterfly}, then $\Re(\mA)$ also admits a Kronecker-sparse factorization with the new architecture $(\vpi_1', \ldots, \vpi_L')$ where $\vpi_\ell'$ is as in \eqref{eq:new_chain}.
\end{corollary}

The chainability condition allows to obtain error guarantees on the Kronecker-sparse factorizations performed in \cite{le2025butterfly}, which is a significant advantage in practical applications.

A similar result holds for the imaginary part since $\Im(\mK_{1}\ldots \mK_{L}) = \Re((-\jmath \mK_{1})\ldots \mK_{L})$ and $\jmath \mK_{1}$ is Kronecker-sparse with the same pattern as $\mK_{1}$. In fact, the same holds for $\Re(z\mK_{1}\ldots\mK_{L})$ for any $z \in \bbC$.

A more concrete consequence is the existence of a fully real-valued Kronecker-sparse implementation of the DCT, DST, and DHT when the dimension is a power of two. It is indeed known \cite{le2025butterfly} (and easy to check) that the square dyadic butterfly factors appearing in the FFT (and the monarch factors) are a chainable architecture, and one can express \cite{makhoul1980Fast} the DCT/DST matrices as $\v{L} \Re(\mD \v{F}) \v{R}$, where $\v{L},\v{R}$ are real-valued, $\mD \in \{\mD_{C},\mD_{S}\}$ is diagonal, and $\v{F}$ is a DFT matrix, which can be written as $\v{F} = \mB \v{P}$ with $\v{B}$ a complex matrix with square dyadic butterfly factorization and $\v{P}$ the so-called bit-reversal permutation. As a result the DCT and the DST write as $\v{L} \Re(\mD\mB) \v{PR}$ and we conclude since $\mD \mB$ still has a square dyadic butterfly factorization. Since the DHT is a sum of the DCT and the DST, it similarly writes $\v{L} \Re((\mD_{C}+\mD_{S})\mB) \v{PR}$ which  admits a Kronecker-sparse factorization for the same reasons.

\section{Experiments}

Finally, we perform some experiments to evaluate the inference time of the DCT applied to a vector $\v{x} \in \bbR^{n \times b}$, where $b$ is the batch size, with the real-valued Kronecker-sparse factorization.
Considering a batch of vectors is relevant in practice, since the DCT is rarely applied to a single vector but rather to large batches of data: for instance, image compression standards such as JPEG apply the DCT to every $8 \times 8$ block of an image \cite{wallace1992jpeg}, and audio coders apply the DCT to every frame of a signal \cite{princen1986analysis}.
The code is available in the supplementary material.

\begin{figure}[!t]
    \centering
    \includegraphics[width=0.9\columnwidth]{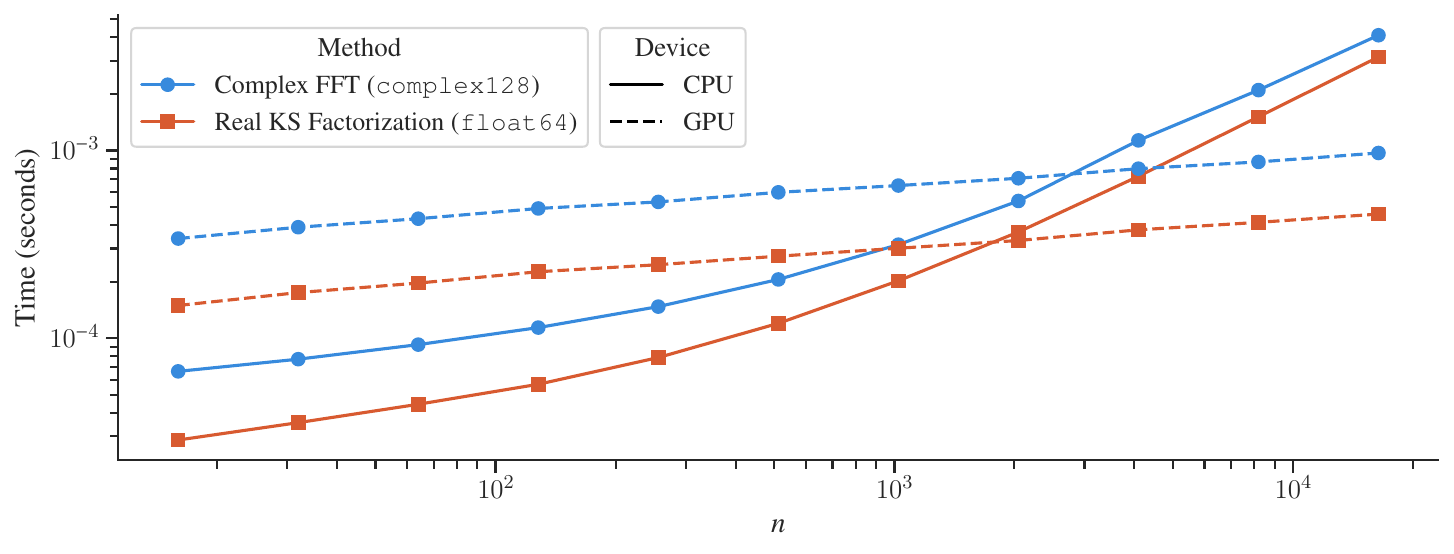}
    \\
    \includegraphics[width=0.9\columnwidth]{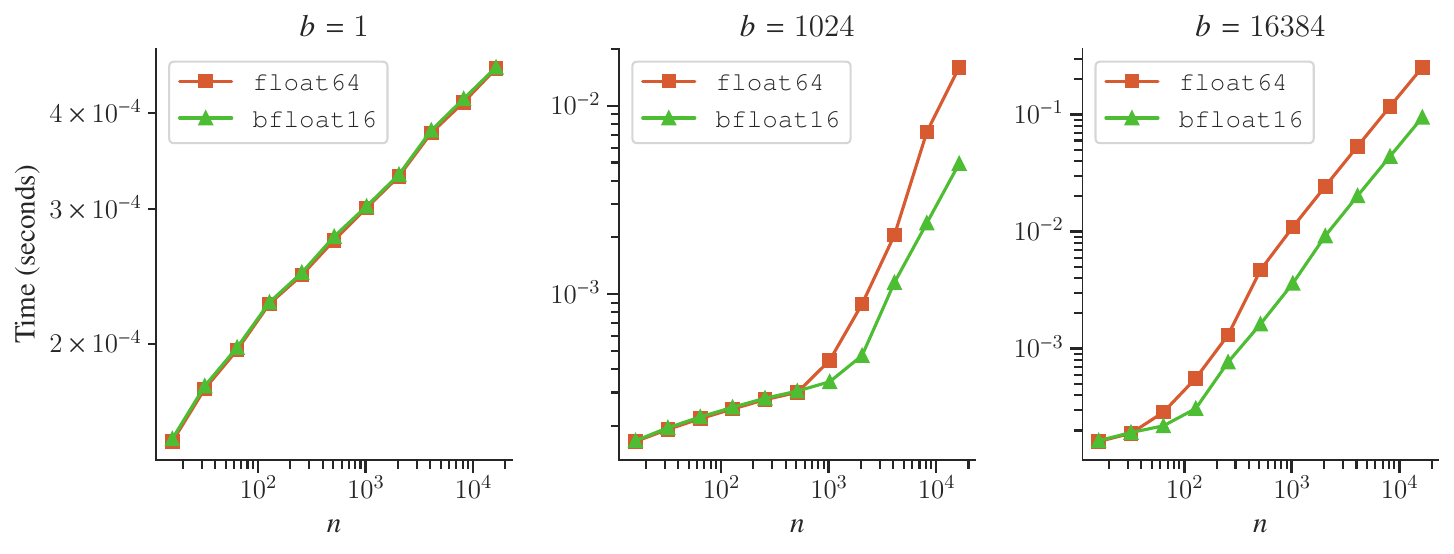}
    \caption{Inference time of the DCT applied to a vector $\v{x} \in \bbR^{n \times b}$ in terms of $n$.
    ({\em Up}) Comparison, on CPU and GPU, with $b = 1$ of two implementations: the one provided by \cite{makhoul1980Fast} based on the complex-valued FFT and the real-valued factorization proposed in this letter.
    ({\em Down}) Real-valued factorization on GPU for different batch sizes $b$ and different floating-point formats.}
    \label{fig:timingDCT}
\end{figure}

As illustrated in \figurename~\ref{fig:timingDCT} ({\em Up}), on CPU and GPU, we observe that the real-valued factorization is faster for every size of the input (e.g., $2$ times faster on GPU), in terms of computation time, than the one provided by \cite{makhoul1980Fast}, which performs a complex-valued discrete Fourier transform.

Moreover our factorization is based on real-valued matrices which, on GPUs, can be computed faster with standard low-precision formats such as \texttt{float16} or \texttt{bfloat16}.
This difference is more pronounced when the batch size is increased.
Indeed, as illustrated in \figurename~\ref{fig:timingDCT} ({\em Down}), while the inference time in \texttt{bfloat16} is equivalent to the one in \texttt{float64} when $b = 1$, for larger $b$ there is a dimension threshold beyond which the inference time in \texttt{bfloat16} is significantly reduced compared to \texttt{float64}.
Since naive rounding to such lower-precision formats (as done in these experiments) comes with some loss of accuracy, recent works have developed methods to mitigate this impact and improve over naive rounding \cite{gribonval2023scaling}.

\section{Conclusion}
\label{sec:conclusion}

In this letter, we addressed the question of whether the real or the imaginary part of a product of structured-sparse matrices, called Kronecker-sparse factors, retains the same structure.
We provided, under a choice of permutation matrices, a necessary and sufficient condition on the factors of the product to ensure this property.
As a direct consequence, we proved that the Discrete Cosine/Sine/Hartley Transforms all admit a real-valued Kronecker-sparse factorization, which, to the best of our knowledge, is a new result in the literature.
This result is of practical interest and has been implemented in the Python library \texttt{lazylinop}\footnote{\url{https://lazylinop.inria.fr}} for the design of efficient algorithms for these transforms.

\section*{Acknowledgement}

The Lean formalization provided in the supplementary material was generated with the assistance of Claude (Anthropic).

\newpage
\IEEEtriggeratref{14}
\bibliographystyle{IEEEtran}
\bibliography{IEEEabrv,biblio}

\end{document}